\documentclass[journal,onecolumn]{IEEEtran}
\usepackage[normalem]{ulem}
\usepackage{url}
\usepackage{cite}
\usepackage[english]{babel} 
\usepackage{enumerate}
\ifCLASSINFOpdf
  \usepackage[pdftex]{graphicx}
  \DeclareGraphicsExtensions{.pdf,.jpg,.png}
  \graphicspath{C:/Users/Amichai/Desktop/BICA transactions/}
\else
\fi
\usepackage{caption}

\usepackage[utf8]{inputenc} 
\usepackage[T1]{fontenc}
\usepackage{array}
\usepackage{url}
\usepackage{ifthen}
\usepackage{cite}
\usepackage{amssymb}
\usepackage{amsfonts}
\usepackage{dsfont}
\usepackage{tabularx}
\usepackage{graphicx}
\usepackage{xfrac}
\usepackage[cmex10]{amsmath}

\newtheorem{theorem}{Theorem}
\newtheorem{lemma}[theorem]{Lemma}
\newtheorem{proposition}{Proposition}

\newenvironment{proof}[1][Proof]{\begin{trivlist}
\item[\hskip \labelsep {\bfseries #1}]}{\end{trivlist}}

\newcommand{\argmax}{\mathrm{arg}\displaystyle\max}

\newcommand{\qed}{\nobreak \ifvmode \relax \else
      \ifdim\lastskip<1.5em \hskip-\lastskip
      \hskip1.5em plus0em minus0.5em \fi \nobreak
      \vrule height0.75em width0.5em depth0.25em\fi}

\newcommand{\RNum}[1]{\uppercase\expandafter{\romannumeral #1\relax}}

\newcolumntype{M}[1]{>{\centering\arraybackslash}m{#1}}
\newcolumntype{N}{@{}m{0pt}@{}}

\DeclareFontFamily{OT1}{pzc}{}
\DeclareFontShape{OT1}{pzc}{m}{it}{<-> s * [1.10] pzcmi7t}{}
\DeclareMathAlphabet{\mathpzc}{OT1}{pzc}{m}{it}

\begin{document}
\title{Innovation Representation of Stochastic Processes\\ with Application to Causal Inference}

\author{Amichai~Painsky,~\IEEEmembership{Member,~IEEE,}
        Saharon~Rosset and~Meir~Feder,~\IEEEmembership{Fellow,~IEEE}

\thanks{A. Painsky is with the Engineering and Computer Science Department, The Hebrew University of Jerusalem, Israel. Contact: amichai.painsky@mail.huji.ac.il}
\thanks{S. Rosset is with the Statistics Department, Tel Aviv University, Israel}
\thanks{M. Feder is with the Department of Electrical Engineering, Tel Aviv University,  Israel}
\thanks{The material in this paper was presented in part at the IEEE International Symposium on Information Theory (ISIT) 2013 \cite{painsky2013memoryless}.}

\thanks{This research was funded in part by Israeli Science Foundation grant 634-09 and by a grant to Amichai Painsky from the Israeli Center for Absorption in Science}
}

%
%

\markboth{IEEE TRANSACTIONS ON INFORMATION THEORY}%
{Shell \MakeLowercase{\textit{et al.}}: Bare Demo of IEEEtran.cls for Journals}
%



\maketitle

\begin{abstract}
Typically, real-world stochastic processes are not easy to analyze. In this work we  study the representation of any stochastic process as a memoryless innovation process triggering a dynamic system. We show that such a representation is always feasible for innovation processes taking values over a continuous set. However, the problem becomes more challenging when the alphabet size of the innovation is finite. In this case, we introduce both lossless and lossy frameworks, and provide closed-form solutions and practical algorithmic methods. In addition, we discuss the properties and uniqueness of our suggested approach. Finally, we show that the innovation representation problem has many applications. We focus our attention to Entropic Causal Inference, which has recently demonstrated promising performance, compared to alternative methods. 
\end{abstract}


%
\IEEEpeerreviewmaketitle

\section{Introduction}

\IEEEPARstart{C}{onsider} a time-dependent stochastic process $X^k=X_1,\dots,X_k$, and a corresponding realization $x^k=x_1,\dots,x_k$.  In this work we study an \textit{innovation representation} problem of the form 
\begin{equation}
\label{basic}
X_k \approx g(Y_k,x^{k-1}) 
\end{equation}
where $X^k$ is to be accurately described by a memoryless (independent over time) process $Y^k$ that triggers a deterministic system, $g(\cdot,\cdot)$.  We refer to $Y^k$ as the innovation process of $X^k$, as it summarizes all the new information that is injected to the process at time $k$. Therefore, the representation problem in \ref{basic} is equivalent to a sequential reconstruction of a memoryless process $Y^k$ from a given process $X^k$.

Over the years, several methods have been introduced for related problems. The Gram-Schmidt procedure \cite{arfken1985gram} suggests a simple sequential method which projects every new component on the linear span of the components that where previously observed. The difference between the current component and its projection is guaranteed to be orthogonal to all previous components. Applied to a Gaussian process, orthogonality implies statistical independence and the subsequent process is therefore considered memoryless. On the other hand, non-Gaussian processes do not hold this property and a generalized form of generating a memoryless process from any given time dependent series is required. Several non-sequential methods such as Principal Components Analysis \cite{jolliffe2002principal} and Independent Component Analysis \cite{hyvarinen1998independent,painsky2015generalized,painsky2018linear} have received a great deal of attention, but we are aware of a little previous work on sequential schemes for generating memoryless innovation processes. 

 The importance of innovation representation spans a variety of fields. One example is dynamic system analysis in which complicated time dependent processes are approximated as independent processes triggering a dynamic system (human speech mechanism, for instance). Another common example is cryptography, where a memoryless language is easier to encrypt as it prevents an eavesdropper from learning the code by comparing its statistics with those of the serially correlated language.
In Communications, Shayevitz and Feder presented the Posterior Matching (PM) principle \cite{shayevitz2011optimal}, where an essential part of their scheme is to produce statistical independence between every two consecutive transmissions. Recently,  innovation representation was further applied to causal inference \cite{kocaoglu2017entropic}. Here, given two variables $X,Y$, we say that $X$ causes $Y$ if $Y$ can be represented as mapping of $X$ and an additional ``small" independent innovation variable $E$, i.e. $Y\approx g(X,E)$. We discuss causal inference in detail in Section \ref{section_causal}. 

In this work we address the innovation representation problem in a broad perspective and introduce a general framework to construct memoryless processes from any given time-dependent process, under different objective functions and constraints. 

\section{Problem Formulation}
Let  $X^k$ be a random process, described by its cumulative distribution function $F(X^k)$. As mentioned above, we would like to construct $Y^k$ such that:
\begin{enumerate}[(a)]
\item	$F(Y^k)=\prod_{i=1}^k F(Y_i)$
\item	$X^k$ can  be uniquely recovered from $Y^k$ for any $k$. \label{uniquely_recover}
\end{enumerate}
In other words, we are looking for a sequential invertible transformation on the set of random variables  $X^k$, so that the resulting variables $Y^k$ are statistically independent.

Interestingly, we show that the two requirements can always be satisfied if we allow $Y_k$ to take values over a continuous set. However, this property does not always hold  when $Y_k$ is restricted to take values over a finite alphabet and need to be relaxed in the general case.  We discuss the continuous case in the next section, followed by a discussion on the discrete case in the remaining sections of this manuscript. 

\section{Innovation Representation for Continuous Variables}

Following the footsteps of the Posterior Matching scheme \cite{shayevitz2011optimal}, we define a generalized Gram-Schmidt method in the continuous case.  

\begin{theorem}
\label{seq_theorem_1}
Let $X \sim F_X (x)$ be a random variable and $\theta \sim \text{Unif}[0,1]$ be statistically independent of it. In order to shape $X$ to a uniform distribution (and vice versa) the following applies:
\begin{enumerate}
\item	$F_X^{-1}(\theta) \sim F_X(x)$
\item	Assume $X$ is a non-atomic distribution ($F_X (x)$ is strictly increasing) then  $F_X(X)\sim \text{Unif}[0,1]$
\item	Assume $X$ is discrete or a mixture probability distribution then  $F_X(X)-\theta P_X(X) \sim \text{Unif}[0,1]$, where $P_X(x)$ is the probability mass at the point $x$. 
\end{enumerate}
\end{theorem}
A proof for this theorem is provided in Appendix $1$ of \cite{shayevitz2011optimal}.

Going back to our problem, define $\tilde{F}_X (x)$ as $\tilde{F}_X (x)=F_X (x)$ if $F_X (x)$ is strictly increasing and $\tilde{F}_X (x)=F_X (x)-\theta P_X (x)$ otherwise. For a desired $F_{Y_k}(y_k)$ we construct our process by setting:
\begin{equation}
Y_1=F_{Y_1}^{-1} \left(\tilde{F}_{X_1} (X_1)\right)
\end{equation}
\begin{equation}
Y_k=F_{Y_k}^{-1} \left(\tilde{F}_{X_k |X^{k-1}} (X_k |x^{k-1} )\right) \quad \forall k>1
\end{equation}
Theorem \ref{seq_theorem_1} guarantees that $\tilde{F}_{X_k |X^{k-1}} \left(X_k |x^{k-1} \right)$ is uniformly distributed and applying $F_{Y_k}^{-1}$ to it shapes it to the desired continuous distribution, $F(Y_k)$. In other words, this method suggests that for every possible history of the process at a time $k$, the transformation $\tilde{F}_{X_k |X^{k-1}} \left(X_k |x^{k-1}\right)$ shapes $X_k$ to the same (uniform) distribution. This ensures independence of its history. The method then reshapes it to the desired distribution. It is easy to see that $Y_k$ are statistically independent as every $Y_k$ is independent of $X^{k-1}$. Moreover, since $F(Y_k)$ is strictly increasing and $\tilde{F}_{X_1}(X_1)$ is uniformly distributed we can uniquely recover $X_1$ from $Y_1$ according to the construction of Theorem \ref{seq_theorem_1}. Simple induction steps show that this is correct for every $Y_k$ for  $k>1$. 

Interestingly, we show (Appendix  A) that this scheme is unique for all monotonically increasing transformations $Y_k=g(X_k,x^{k-1})$. Moreover, any non-increasing transformation that satisfies the requirements above is necessarily a measurable permutation of our scheme. A detailed discussion on these uniqueness properties is provided in Appendix A.

\section{Innovation Representation for Discrete Variables - the Lossy Case}
 Let us now assume that both $X_k$ and $Y_k$ take values on finite alphabet size of $A$ and $B$ respectively (for every $k$). Even in the simplest case, where both are binary and $X$ is a first order non-symmetric Markov chain, it is easy to see that no transformation can meet both of the requirements mentioned above. We therefore relax our problem by replacing the uniquely recoverable requirement (\ref{uniquely_recover}) with mutual information maximization of  $I\left(X_k;Y_k |X^{k-1} \right)$. This way, we make sure that the mutual information between the two processes is maximized at any time given its history. Mutual information maximization is a well-established criterion in many applications; it is easy to show that it corresponds to minimizing the logarithmic loss, which holds many desirable properties \cite{painsky2018universality,painsky2018bregman}. Notice that the case where $X_k$ is uniquely recoverable from $Y_k$ given its past, results in $ I\left(X_k;Y_k |X^{k-1} \right)$ achieving its maximum as desired.

Our problem is reformulated as follows:
for any realization of $X_k$, given any possible history the process $X^{k-1}$, find a set of mapping functions to a desired distribution $P(Y_k)$ such that the mutual information between the two processes is maximal. For example, in the binary case where $X_k$ is a first order Markov process, and $Y_k$ is i.i.d. Bernoulli distributed,
\begin{equation}
\label{markov_model}
Y_k \sim \text{Ber}(\beta_k), \quad P_{X_k}(X_k=0)=\gamma_k
\end{equation}
\begin{equation}
\nonumber
P_{X_k|X_{k-1}}\left(X_k=0|X_{k-1}=0\right)=\alpha_1
\end{equation}
\begin{equation}
\nonumber
P_{X_k|X_{k-1}}\left(X_k=0|X_{k-1}=1\right)=\alpha_2
\end{equation}
we would like to maximize
\begin{equation}
 I\left(X_k;Y_k | X^{k-1}\right)=\gamma_{k-1} I\left(X_k;Y_k | X_{k-1}=0\right)+(1-\gamma_{k-1}) I\left(X_k;Y_k | X_{k-1}=1\right).
\end{equation}
In addition, we would like to find the distribution of $Y_k$ such that this  mutual information is maximal. This distribution can be viewed as the closest approximation of the process $X^k$ as a memoryless process in terms of maximal mutual information with it. 
Notice that this problem is a concave minimization over a convex polytope-shaped set \cite{kovacevic2012hardness}, and the maximum is guaranteed on to lie on one of the polytope's vertices. Unfortunately, this problem is hard and generally there is no closed-form solution to it. Several approximations and exhaustive search solutions are available for this kind of problems, such as \cite{kuno2007simplicial}. However, there are several simple cases in which a closed-form solution exists. One notable example is the binary case.

\subsection{The Binary Case}
\label{binary_case}
Let us first consider the following problem: given two binary random variables $X$ and $Y$ and their marginal distributions $P_X (X=0)=\alpha<\frac{1}{2}$ and $P_Y (Y=0)=\beta<\frac{1}{2}$ we would like to find the conditional distributions $P_{Y|X} (y|x)$ such that the mutual information between $X$ and $Y$ is maximal. Simple derivation shows that the maximal mutual information is: \\
For $\beta>\alpha$:
\begin{equation}
\label{beta>alpha}
I_{\text{max}}^{\beta>\alpha}(X;Y)=h_b(\beta)-(1-\alpha)h_b\left(\frac{\beta-\alpha}{1-\alpha} \right).
\end{equation}
For $\beta<\alpha$:
\begin{equation}
\label{beta<alpha}
I_{\text{max}}^{\beta<\alpha}(X;Y)=h_b(\beta)-\alpha h_b\left(\frac{\beta}{\alpha} \right).
\end{equation}
Applying these results to the first order Markov process setup described above, and assuming all parameters are smaller than $\frac{1}{2}$, we get that the maximal mutual information is simply:\\

\noindent For $\beta_k<\alpha_1<\alpha_2$:
\begin{equation}
 I\left(X_k;Y_k | X^{k-1}\right)=\gamma_{k-1} I_{\text{max}}^{\beta<\alpha_1}\left(X;Y\right)+(1-\gamma_{k-1}) I_{\text{max}}^{\beta<\alpha_2}\left(X;Y\right).
\end{equation}

\noindent For $\alpha_1 \leq \beta_k <\alpha_2$:
\begin{equation}
 I\left(X_k;Y_k | X^{k-1}\right)=\gamma_{k-1} I_{\text{max}}^{\beta>\alpha_1}\left(X;Y\right)+(1-\gamma_{k-1}) I_{\text{max}}^{\beta<\alpha_2}\left(X;Y\right).
\end{equation}

\noindent For $\alpha_1 <\alpha_2 \leq \beta_k $:
\begin{equation}
 I\left(X_k;Y_k | X^{k-1}\right)=\gamma_{k-1} I_{\text{max}}^{\beta>\alpha_1}\left(X;Y\right)+(1-\gamma_{k-1}) I_{\text{max}}^{\beta>\alpha_2}\left(X;Y\right).
\end{equation}

It is easy to verify that $ I\left(X_k;Y_k | X^{k-1}\right)$ is continuous in $\beta_k$. Simple derivation shows that for $\beta_k<\alpha_1<\alpha_2$ the maximal mutual information is monotonically increasing in $\beta_k$ and for $\alpha_1 <\alpha_2 \leq \beta_k $ it is monotonically decreasing in $\beta_k$. It can also be verified that all optimum points in the range of  $\alpha_1 \leq \beta_k <\alpha_2$ are local minima which leads to the conclusion that the maximum must be on the boundary of the range, $\beta_k=\alpha_1$ or $\beta_k=\alpha_2$.  The details of this derivation are located in Appendix B. For example, Figure \ref{illustration} illustrates the shape of  $ I\left(X_k;Y_k | X^{k-1}\right)$ as a function of $\beta_k$, for $\alpha_1=0.15$, $\alpha_2=0.45$.

Since we are interested in the $\beta_k$ that maximizes the mutual information between the two possible options, we are left with a simple decision rule
\begin{equation}
\label{decision_rule}
\gamma_{k-1}\begin{array}{c} 
\beta_k=\alpha_2\\ \lessgtr \\\beta_k=\alpha_1\end{array} \frac{h_b(\alpha_2)-h_b(\alpha_1)+\alpha_2 h_b\left(\frac{\alpha_1}{\alpha_2}\right)}{\alpha_2 h_b \left(\frac{\alpha_1}{\alpha_2}\right)+(1-\alpha_1)h_b \left(\frac{\alpha_2-\alpha_1}{1-\alpha_1}\right)}
\end{equation}
which determines the conditions according to which we choose our $\beta_k$, depending on the parameters of the problem $\gamma_{k-1}, \alpha_1, \alpha_2$.\\

\begin{figure}[h]
\centering
\includegraphics[width = 0.4\textwidth,bb= 80 440 410 710,clip]{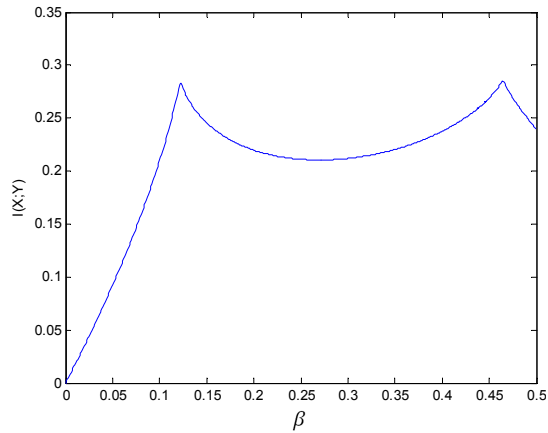}
\caption{$ I\left(X_k;Y_k | X^{k-1}\right)$ as a function of $\beta_k$, for a first order Markov model (\ref{markov_model}), with $\alpha_1=0.15$, $\alpha_2=0.45$}
\label{illustration}
\end{figure}

Further, assuming that the process $X$ is at its stationary state yields that $\beta_k$ is fixed for every $k$ and $\gamma=\frac{\alpha_2}{1-\alpha_1+\alpha_2}$. Applying this result to the decision rule above (\ref{decision_rule}), it is can be verified (Appendix B) that for $\alpha_1<\alpha_2<\frac{1}{2}$  we have:
\begin{equation}
\nonumber
\frac{\alpha_2}{1-\alpha_1+\alpha_2}< \frac{h_b(\alpha_2)-h_b(\alpha_1)+\alpha_2 h_b\left(\frac{\alpha_1}{\alpha_2}\right)}{\alpha_2 h_b \left(\frac{\alpha_1}{\alpha_2}\right)+(1-\alpha_1)h_b \left(\frac{\alpha_2-\alpha_1}{1-\alpha_1}\right)}
\end{equation}
which leads to the conclusion that $\beta_{opt}=\alpha_2$.

The derivation above is easily generalized to all values of $\alpha_1$ and $\alpha_2$. This results in a decision rule stating that $\beta_{opt}$ equals the parameter closest to $\frac{1}{2}$:
\begin{equation}
\beta_{opt}=\argmax_{\theta \in \{\alpha_1,\alpha_2,1-\alpha_1,1-\alpha_1\}}\left(\frac{1}{2}-\theta\right).
\end{equation}
In other words, in order to best approximate a binary first order Markov process at its stationary state we set the distribution of the binary memoryless process to be similar to the conditional distribution which holds the largest entropy. Generalizing this result to an $r$-order Markov process we have $R=2^r$ Bernoulli distributions to be mapped to a single one, $P(Y_k)$. The maximization objective is therefore  
\begin{equation}
 I\left(X_k;Y_k | X^{k-1}\right)=\sum_{i=0}^{R-1} \gamma_i I\left( X_k;Y_k | \left[ X_{k-1}\, \dots \, X_{k-R-1}\right]^T=i\right)
\end{equation}
where  $\gamma_i \triangleq P\left(\left[X_{k-1}\,\dots \,X_{k-R-1}\right]^T=i\right)$. Notice that $I\left( X_k;Y_k | \left[ X_{k-1}\, \dots \, X_{k-R-1}\right]^T=i\right)$ is either $h_b(\beta)-\alpha_i h_b \left(\frac{\beta}{\alpha_i}\right)$ or $h_b(\beta)-(1-\alpha_i) h_b \left(\frac{\beta-\alpha_i}{1-\alpha_i}\right)$, depending on $\beta$ and $\alpha_i$,  as described in (\ref{beta>alpha}--\ref{beta<alpha}). Simple calculus shows that as in the $R=2$ case, the mutual information $ I\left(X_k;Y_k | X^{k-1}\right)$ reaches its maximum on one of the inner boundaries of $\beta$'s range
\begin{equation}
\beta_{opt}=\argmax_{\beta \in \{\alpha_j\}} \left(h_b(\beta)-\sum_{\beta<\alpha_j} \gamma_j \alpha_j h_b\left(\frac{\beta}{\alpha_j}\right)-\sum_{\beta>\alpha_j} \gamma_j (1-\alpha_j) h_b\left(\frac{\beta-\alpha_j}{1-\alpha_j}\right)\right).
\end{equation}
Unfortunately, here it is not possible to conclude that $\beta$ equals the parameter closest to $\frac{1}{2}$, as a result of the nature of our concave minimization problem.  

\section{Innovation Representation for Discrete Variables - the Lossless Case}
\label{section_lossless}
The lossy approximation may not be adequate in applications where unique recovery of the original process is required. It is therefore necessary to increase the alphabet size of the output so that every marginal distribution of $X_k$, given any possible history of the process, is accommodated.  
This problem can be formulated as follows: 

Assume we are given a set of $R$ random variables, $\{X_i\}_{i=1}^R$, such that each random variable $X_i$ is multinomial distributed, taking on $A$ values,  $X_i \sim \text{multnom}\left(\alpha_{1i},\alpha_{2i},\dots,\alpha_{Ai}\right)$. Notice that $X_i$ corresponds to $X_k|[X_{k-1},\dots,X_{k-R-1}]^T=i$, as appears in the previous section, and $A$ is the marginal alphabet size of the original process $X_k$. Using the notation from previous sections, we have that $P(X_i)$ corresponds to $P(X_k | \left[ X_{k-1}\, \dots \, X_{k-R-1}\right]^T=i)$. In addition, we denote $x_{a}$ as the $a^{th}$ symbol of random variable $X_i$. We would like to find a distribution $Y \sim \text{multnom}\left(\beta_1,\beta_2,\dots,\beta_B\right)$  where the $\beta$'s and alphabet size $B \geq A$ are unknown. In addition, we are looking for $R$ sets of conditional probabilities between every possible realization $X_i=x_{a}$ and $Y$, such that $X_i=x_{a}$ can be uniquely recoverable from $Y=y_b$, for every $i,a$ and $b$. Further, we would like the entropy of $Y$ to be as small as possible so that our memoryless process is as ``cheap" as possible to describe. Since the transformation is invertible, we have that $H(X_k|X^{k-1}=x^{k-1})=H(Y_k)$, for every $k$. This means that by minimizing $H(Y_k)$ we actually minimize $H(X_k|X^{k-1}=x^{k-1})$. In other words, our problem may be viewed as follows: given a set of $R$ marginal distributions $X_i \sim \text{multnom}\left(\alpha_{1i},\alpha_{2i},\dots,\alpha_{Ai}\right)$ for $i={1,\dots,R}$, we are looking for a joint probability distribution such that the joint entropy is minimal.  This problem was recently introduced as \textit{Minimum Entropy Coupling} \cite{kocaoglu2017entropic} and shown to be NP hard.

Without loss of generality we assume that $\alpha_{ai} \leq \alpha_{(a+1)i}$ for all $a \leq A$, since we can always order them this way. We also order the sets according to the smallest parameter,  $\alpha_{1i} \leq \alpha_{1(i+1)}$. Notice we have $\alpha_{1i}\leq \frac{1}{2}$ for all $i=1,\dots, R$ , as an immediate consequence. 

 For example, for $A=2$ and $R=2$, it is easy to verify that $B \geq 3$ is a necessary condition for $X_i$ to be uniquely recoverable from $Y$. Simple calculus shows that the conditional probabilities which achieve the minimal entropy are $\beta_1=\alpha_1,\,  \beta_2=\alpha_2-\alpha_1$ and $\beta_3=1-\alpha_2$, as appears in Figure \ref{lossless_matching} for $\alpha_1\leq \alpha_2 \leq \frac{1}{2}$. 

\begin{figure}[h]
\centering
\includegraphics[width = 0.75\textwidth,bb= 50 520 540 720,clip]{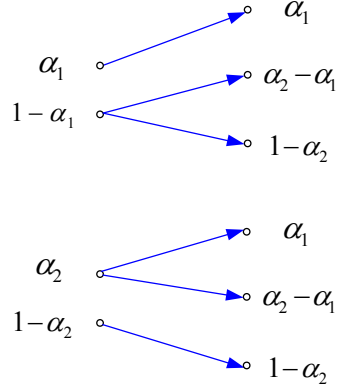}
\caption{Optimal lossless representation of two binary sources with a single ternary source}
\label{lossless_matching}
\end{figure}

\subsection{Minimizing $B$}
\label{minimizing_b_section}
Let us start by finding the minimal alphabet size of the output process $B$, such that   $X^k$ is guaranteed to be uniquely recoverable from it.
Looking at the free parameters of our problem, we first notice that defining the distribution of $Y$ requires exactly $B-1$ parameters. Then, defining $R$ conditional probability distributions between each alphabet size $A$ and the output process $Y$ takes $R(A-1)(B-1)$ parameters. For $X_i$ to be uniquely recoverable from $Y$, each value of $Y$ needs to be assigned to at most a single value of $X_i$ (see Figure \ref{lossless_matching} for example). This means that for each of the $R$ sets, we have $B(A-1)$ constraints ($B$ possible realizations of $Y$, each of them has $A-1$ zero conditional probability constraints). Therefore, in order to have more free parameters than constraints we require that:
\begin{equation}
(B-1)+R(A-1)(B-1) \geq RB(A-1).                          
\end{equation}
Rearranging this inequality leads to
\begin{equation}
B \geq R(A-1)+1.                    	
\end{equation}

For example, assuming the $X_i$ takes over a binary alphabet, we get that $B\geq R+1$. There exist several special cases in which it is possible to go under this lower bound, like cases where some parameters are additions or subtraction of other parameters. For example, $\alpha_2=1-\alpha_1$ in the binary case. Our derivation focuses on the most general case. Notice that a similar result appears in Lemma 3 of \cite{kocaoglu2017entropic}. However, it is important to emphasize that an earlier conference version of our results was already published in \cite{painsky2013memoryless} and \cite{painsky2018phd}, several years before \cite{kocaoglu2017entropic}. 

\subsection{The Optimization Problem}
The problem above can be formulated as the following optimization problem:
\begin{equation}
\min H(Y) \quad \text{s.t.} \quad H(X_k | Y = y_b,\; [X_{k-1},\dots,X_{k-R-1}]^T=i) = 0 \quad \forall i =1,\dots, R,\; b=1,\dots,B
\end{equation}
Unfortunately this is a concave minimization problem over a non-convex set. However, we show that this problem can also be formulated as a mixed integer problem.

\subsection{Mixed Integer Problem Formulation}
\label{mixed_int}
In order to formulate our problem as a mixed integer problem we first notice that the free parameters are all conditional probabilities, as they fully determine the outcome distribution. We use the notation $p_{iab}$ to describe the conditional probability $P(Y=y_b | X_i=x_{a})$. The equality constraints we impose on our minimization objective are as follows:

\begin{itemize}
\item		All $R$ conditional probability sets must result in the same output distribution:
\begin{align}
\nonumber
P(Y=y_b)=&\sum_{a=1}^A P\left(Y=y_b | X_i=x_{a}\right) P\left( X_i=x_{a}\right)=\sum_{a=1}^A p_{iab} \alpha_{ai}
\end{align}
for all $i=1,\dots,R$ and $b=1,\dots,B$. Since the parameters $\alpha_{1i},\dots,\alpha_{Ai}$ are given, we have that 
\begin{align}
\nonumber
\sum_{a=1}^A p_{iab}\alpha_{ai}-\sum_{a=1}^A p_{jab}\alpha_{aj}=0
\end{align}
for all $i,j=1,\dots,R$ and $b=1,\dots,B.$

\item	$P(Y|X_i)$ is a valid conditional distribution function:
\begin{align}
\nonumber
\sum_{b=1}^A p_{iab}=1 
\end{align}
for all $i=1,\dots,R$ and $a=1,\dots,A.$
\end{itemize}

In addition, the inequality constraints are:
\begin{itemize}
\item	For convenience, we ask that $P(Y=y_b ) \leq P(Y=y_{b+1} )$ for all $b=1,\dots,B$:
\begin{align}
\nonumber
\sum_{a=1}^A p_{iab}\alpha_{ai}-\sum_{a=1}^A p_{ia(b+1)}\alpha_{ai} \leq 0 \quad \text{for all} \quad 1 \leq b \leq B.
\end{align}

\item Zero conditional entropy constraint:
as stated above, a necessary and sufficient condition for zero conditional entropy is that for every value $Y=y_b$, in every set $i=1,\dots,R$, there is only a single value $X_i=x_{a}$ such that $p_{iab}>0$.
Therefore, for each of the $R$ sets, and for each of the $B$ values $Y$ can take on, we define $A$ boolean variables, $T_{iab}$, that  satisfy:
\begin{equation}
\nonumber
p_{iab}-T_{iab} \leq 0, \quad \sum_{a=1}^A T_{iab} =1, \quad T_{iab} \in \{0,1\}.
\end{equation}
Notice that the summation ensures only a single $T_{iab}$ equals one, for which $p_{iab}\leq1$. For each of the other $T_{iab}=0$ the inequality constraint verifies that $p_{iab}\leq0$.
This set of constraints can also be written using $A-1$ Boolean variables:
\begin{equation}
\nonumber
p_{iab}-T_{iab} \leq 0 \quad  \forall \,\, a=1,\dots,A
\end{equation}
\begin{equation}
\nonumber
p_{iAb}-\left(1-\sum_{a=1}^{A-1}T_{iab}\right) \leq 0 \quad \Leftrightarrow \quad p_{iAb}+\left(\sum_{a=1}^{A-1}T_{iab}\right) \leq 1 
\end{equation}
\begin{equation}
\nonumber
T_{iab} \in \{0,1\} \quad \forall \,\, a=1,\dots,A.
\end{equation}

\end{itemize}

Therefore, our minimization problem can be written as follows:
define a vector of parameters $z=\left[p_{iab} \quad T_{iab}\right]^T$. 
Define $A_{eq}$ and $b_{eq}$ as the equality constraints in a matrix and vector forms respectively.
Define $A_{ineq}$ and $b_{ineq}$ as the inequality constraints in a matrix and vector forms respectively. This leads to
\begin{equation}
\label{f(z)}
\min f(z) 
\end{equation}
\begin{equation}
\nonumber
\text{s.t.} \quad A_{eq}z=b_{eq}
\end{equation}
\begin{equation}
\nonumber
\quad \quad \quad \;\; A_{ineq}z\leq b_{ineq}
\end{equation}
\begin{equation}
\nonumber
\quad   \;\; 0\leq z \leq 1
\end{equation}
\begin{equation}
\nonumber
\quad    \quad  \quad  \quad \quad \quad \;\; T_{iab} \in \{0,1\} \quad \forall i,a,b
\end{equation}
where f(z) is the entropy of the random variable $Y$ in terms of $p_{iab}$ and \textit{boolean indicators} define which elements in $z$ correspond to $T_{iab}$.
Mixed integer problems are studied broadly in the computer science community. There are well established methodologies for convex minimization in a mixed integer problem and specifically in the linear case \cite{floudas1995nonlinear,tawarmalani2004global}. The study of non-convex optimization in mixed integer problems is also growing quite rapidly, though there is less software available yet. The most broadly used mixed integer optimization solver is CPLEX, developed by IBM. CPLEX provides a mixed integer linear programming (MILP) solution, based on a branch and bound oriented algorithm. We use the MILP in lower bounding our objective function (\ref{f(z)}) as described in the following sub-sections.

\subsection{Greedy Solution}
\label{greedy_section}
The entropy minimization problem can also be viewed as an attempt to minimize the entropy of a random variable $Y \sim \text{multinom}\left(\beta_1,\beta_2,\dots,\beta_B \right)$ on a set of discrete points representing valid solutions to the problem we defined. 
Let us remember that $\beta_b\leq \beta_{b+1}$ for all $b=1,\dots,B$ as stated in the previous sections. 
\begin{proposition}
\label{propprop}
$\beta_B$ is not greater than $\min_i\{ \alpha_{Ai}\}$
\end{proposition}
\begin{proof}
Assume $\beta_B>\min_i\{ \alpha_{Ai}\}$. Then, for this $i$ there must be at least two values $x_{u}$ and $x_{v}$ for which $p_{iub}>0$ and $p_{ivb}>0$. This contradicts the zero conditional entropy constraint. \hfill $\square$
\end{proof}
For example, Figure \ref{lossless_matching} demonstrates the optimal lossless solution for $R=2$, and $B=3$. We see that $\beta_B=\min\{1-\alpha_1, 1-\alpha_2\}$ (for $\alpha_1\leq \alpha_2 \leq \frac{1}{2}$). Moreover, it is easy to verify that $\beta_B=1-\alpha_2$ is not a feasible solution, as Proposition \ref{propprop} suggests.

Therefore, a greedy algorithm would  like to “squeeze” all the distribution to the values which are less constrained from above, so that it is as large as possible. 

Our suggested algorithm works as follows: first set $B$ according to the bound presented in Section \ref{minimizing_b_section}. Then, in every step of the algorithm we set $\beta_B=\min_i\{ \alpha_{Ai}\}$. This leaves us with a $B-1$ problem (of setting the remaining values of $\beta_1,\dots,\beta_{B-1}$ ). Given this mapping, we rearrange the remaining probabilities $\alpha_{ai}$ and repeat the previous step. We terminate once we set the smallest value, $\beta_1$. This process ensures that in each step we increase the least constrained  value of $\beta$ as much as possible. 

We notice that the same greedy algorithm was recently introduced and studied by Kocaoglu et al. \cite{kocaoglu2017entropic} in the context of causal inference. However, as stated in Section \ref{section_lossless}, an earlier conference version of our results, including this algorithm, was already published in \cite{painsky2013memoryless} (and in a more detail in \cite{painsky2018phd}) several years before \cite{kocaoglu2017entropic}.

\subsection{Lowest Entropy Bound}
\label{lowest_section}
As discussed in the previous sections, we are dealing with an entropy minimization problem over a discrete set of valid solutions. Minimizing the entropy over this set of points can be viewed as a mixed integer non-convex minimization, which is a hard problem. 
In this section we introduce a lower bound to the optimal solution by relaxing the domain of solutions to a continuous set. In other words, instead of searching for $\beta$'s over a discrete set (as a result of the requirement for $R$ invertible mappings), we now search for $\beta$'s over a continuous set, which naturally includes additional unfeasible solutions.  We would like to consider the smallest continuous set that includes all feasible solutions. For this purpose, we find an upper and lower bound for each of the parameters  $\beta_b$ and solve the problem when $\beta_b$ may take any value in this range. This way, we relax the search over a set of valid solutions to a search in a continuous space, bounded by a polytope, and attain a lower bound to the desired entropy. 

We find the boundaries for each $\beta_b$ by changing our minimization objective to a simpler linear one (minimize/maximize $\beta_b$, for every $b=1,\dots,B$ at a time). This problem is a simple MILP as shown above.  By looking at all these boundaries together we may minimize the entropy in this continuous space $\beta_{b,min}\leq \beta_b \leq \beta_{b,max}$ and find a lower bound for the minimal entropy one can expect. Theorem \ref{theorem_lb} states the constructive conditions for an optimal solution in this case. 

We notice that this bound is not tight, and we even do not know how far it is from a valid minimum, as it is not necessarily a valid solution. However, it gives us a benchmark to compare our greedy algorithm against and decide if we are satisfied with it, or require more powerful tools. We also note that as $B$ increases, the number of valid solutions grows exponentially. This leads to a more packed set of solutions which tightens the suggested lower bound as we converge to a polytope over a continuous set.

\begin{theorem}
\label{theorem_lb}
Let $Y$ be a random variable over multinomial distribution, $Y\sim \text{multnom}(\beta_1,\beta_2,\dots,\beta_B )$. Assume that parameters $\beta_i$ satisfy:
\begin{enumerate}
\item $a_i \leq \beta_i \leq b_i$ for all $i=1,\dots,B$.
\item $\sum a_i \leq 1$ and $\sum\beta_i \geq 1$ (to ensure the existence of a feasible solution).
\item $a_i\leq a_{i+1}$ and   $b_i\leq b_{i+1}$ for all $i=1,\dots,B$.
\end{enumerate}
Then, the minimal entropy is achieved by 
\begin{enumerate}
\item $\beta_i =b_i$ for all $i>k$.
\item $\beta_i =a_i$ for all $i<k$.
\item $\beta_k =1-\sum_{i\neq k} \beta_i$.
\end{enumerate}
for some $k>0$. 
\end{theorem}
A proof for this theorem is provided in Appendix C.

\section{Applications}
As mentioned above, the innovation representation problem has many applications in a variety of fields. In this section we focus on two important applications. We begin with causal inference, as recently introduced by kocaoglu et al. \cite{kocaoglu2017entropic}.

\subsection{Causal Inference}
\label{section_causal}
Consider two random variables $X$ and $Y$. We say that $X$ \textit{causes} $Y$ (denote as $X\rightarrow Y$) if a change in the value of $Y$ is a consequence of a change in the value of $X$ (and vice versa). A general solution to the causal inference problem is to conduct experiments, also called interventions (for example, \cite{shanmugam2015learning}). For many problems, it can be very difficult to create interventions since they require additional experiments after
the original data-set was collected. Nevertheless, researchers would still like to discover causal relations between variables using only observational data, using so-called data-driven causality. 
However, a fundamental problem in this approach is the symmetry of the underlaying distribution; the joint distribution $P(x,y)$ may be factorized as either $P(x)P(y|x)$ or $P(y)P(x|y)$. This means that we cannot infer the causal direction directly from the joint distribution, and additional assumptions must be made about the mechanisms that generate the data \cite{peters2011causal}. 

The most popular assumption for two-variable data-driven causality is the additive noise model (ANM) \cite{shimizu2006linear}. In ANM, we assume a model $Y = g(X) + E$ where $E$ is a random variable that is statistically independent of $X$. Although restrictive, this assumption leads to strong theoretical guarantees in terms of identifiability, and provides the state of the art accuracy in real datasets. Shimizu et al. \cite{shimizu2006linear} showed that if $g$ is linear and the noise is non-Gaussian, then the causal direction is identifiable.
Hoyer et al. \cite{hoyer2009nonlinear} showed that when $g$ is non-linear, irrespective of the noise, identifiability holds in a non-adverserial setting of system parameters. Peters et al.   \cite{peters2011causal} extended ANM to discrete variables.

Recently, Kocaoglu et al. \cite{kocaoglu2017entropic} extended the ANM framework and introduced the Entropic Causal Inference principle. Specifically, they argue that if the true causal direction is $X\rightarrow Y$ , then the random variable $Y$ satisfies $Y=g(X,E)$ where $g$ is an arbitrary function and $E$ is a ``simple" random variable that is statistically independent of $X$. The ``simplicity" of $E$ is characterized by a low R\'{e}nyi entropy. This means that for any model in the wrong direction, $X = \tilde{g}(Y, \tilde{E})$, the random variable $\tilde{E}$ has a greater Rényi entropy than $E$. Kocaoglu et al. focused on two special case of R\'{e}nyi entropy: $H_0$, which corresponds to the cardinality of $E$, and $H_1$, which is the classical Shannon entropy. They proved an identifiability result for $H_0$, showing that if the probability values are not adversarially chosen, for most functions, the true causal direction is identifiable under their model. Further, they showed that by using Shannon entropy ($H_1$),  they obtain causality tests that work with high probability in synthetic datasets, and slightly outperform state of the art alternative tests in real-world datasets.

The $H_1$ causality test was driven as follows: consider  $Y = g(X,E)$ where $E$ is independent of $X$.  Let $g_x:E \rightarrow Y$ be the mapping from $E$ to $Y$ when $X = x$, i.e., $g_x(E) \triangleq g(x,E)$. Then $P(Y = y|X = x) = P(g_x(E) =y|X = x) = P(g_x(E) = y)$ where the last equality follows from the independence of $X$ and $E$. Thus, the conditional distributions $P(Y |X = x)$ are treated as distributions that emerge by
applying some function $g_x$ to some unobserved variable $E$. Then the problem of identifying $E$ with minimum entropy given the joint distribution $P(x, y)$ becomes equivalent to the following: given distributions of the variables $g_i(E)$, find the distribution with minimum entropy (distribution of $E$) such that there exists functions $g_i$ which map this distribution to the observed distributions of $Y|X = i$. It can be shown that $H(E)  \geq H(g_1(E), g_2(E), \dots, g_R(E))$. Denote $g_i(E)$ as
a random variable $U_i$. Then, the best lower bound on $H(E)$ can be obtained by minimizing $H(U_1, U_2, \dots , U_R)$. Further, it is shown that it is always possible to construct a random variable $E$ that achieves this minimum. Thus the problem of finding the $E$ with minimum entropy given the joint distribution $P(x, y)$ is equivalent to the problem of finding the minimum entropy joint distribution of the random variables $U_i = (Y |X = i)$, given the marginal distributions $P(Y |X = i)$. 

Notice that this problem is equivalent to our lossless representation problem in Section \ref{section_lossless}: Given $R$ distributions (which correspond to the $U_i$'s) , we seek $R$ invertible mappings, from each of the $R$ distributions to $Y$, such that the entropy of $Y$ is minimal. Since the mappings are invertible, we have the $H(Y)=H(U_1, U_2, \dots , U_R)$ and the problem is equivalent to minimizing $H(Y)=H(U_1, U_2, \dots , U_R)$ subject to the marginal distributions of the $U_i$'s.

Kocaoglu et al. \cite{kocaoglu2017entropic} showed that this problem is NP-hard. Further, they conjectured that the  identifiability result they proved for $H_0$ entropy also holds in this case and proposed a greedy algorithm. Interestingly, their suggested algorithm is exactly the same as ours (Section \ref{greedy_section}). However, it is important to emphasize that our algorithm was already introduced in \cite{painsky2013memoryless} (and later in \cite{painsky2018phd}), several years before the work of Kocaoglu et al. \cite{kocaoglu2017entropic}.  
 
In a more recent work, Kocaoglu et al. \cite{kocaoglu2017entropic2} continue the study of the proposed greedy solution. They showed that it converges to a local minimum and derived several algorithmic properties. In addition, they derived a variant of greedy algorithm which is easier to analyze. 

The Entropic Causal Inference principle has gained a notable interest, mostly due its intuitive interpretation and promising empirical results. On the other hand, the proposed greedy algorithm does not guarantee to converge to the optimal solution. Moreover, it is not clear how far it is from the global minimum (or some infimum). Our suggested solutions address these concerns, as described in detail in Sections \ref{mixed_int} and \ref{lowest_section}. 

\subsection{The IKEA Problem}
An additional application of our suggested framework comes from industrial engineering, as it deals with optimal design of mass production storage units.

Consider the following problem: a major home appliances vendor is interested in mass manufacture of storage units. These units hold a single and predetermined design plan according to the market demand. Assume that the customers market is defines by $R$ major storing types (customers) and each of these customers is characterized by a different distribution of items they wish to store. The vendor is interested in designing a single storage unit that satisfies all of his customers. In addition, the vendor would like the storage unit to be as ``compact" and ``cheap" as possible. We refer to this problem as \textit{the IKEA problem}. Consider the $R$ customer distributions $\{X_i\}_{i=1}^R$ such that each customer $X_i$ is over a multinomial distribution with $A$ values, $X_i \sim \text{multnom}(\alpha_{1i},\alpha_{2i},\dots,\alpha_{Ai})$. We assume that all customer distributions have the same cardinality $A$. It is easy to generalize our solution to different cardinalities. We are interested in mapping the $R$ customer distributions into a single storage distribution, $P(Y)$, which represents the storage unit to be manufactured.

First, we would like every customer to be able to store its items exclusively; different items shall not be stored together. As in the previous sections, we use the notation $x_{a}$ to define the $a^{th}$ symbol of the random variable $X_i$. For our storing units problem, we would like to find a multinomial distribution over $B$ values ($B\geq A$ is unknown), $Y \sim \text{multnom}(\beta_1,\beta_2,\dots,\beta_B )$,  and $R$ sets of conditional probabilities between every $X_i=x_{a}$ and $Y$, such that $X_i=x_{a}$ can be uniquely recoverable (reversible) from $Y=y_b$ for every $i,\,a$ and $b$.

In addition, we would like the storing unit to be ``compact" and ``cheap". For most functionalities, a compact storing unit is rectangular shaped (closets, cabins, dressers etc.) and it is made of multiple compartments (shelves) in numerous columns.  We define the number of columns in our storage unit as $L$ and the number of shelves as $N$. 
Therefore, we would like to design a rectangular shaped storing unit such that given a number of columns $L$, every customer is able to store its items exclusively and the number of shelves is minimal. The corresponding technical requirements, in terms of our problem, are discussed in the following sections.

This problem is again NP hard, for the same reasons as in the previous sections, but it can be reformulated to a set of Mixed Integer Quadratic Programming (MIQP) problems, which is an established research area with extensive software available.

\subsection{Mixed Integer Quadratic Programming Formulation}
Let us first assume we are given both the number of columns in our desired storing unit $L$ and the number of shelves $N$. Since we require the storing unit to be rectangular, we need to find such distribution $Y$ that can be partitioned to $L$ columns with no residue. Therefore, we define $L$ equivalent partitions $\{\delta_l\}_{l=1}^L$   in the size of $\frac{1}{L}$ for which each $\{\beta_b\}_{b=1}^B$ is exclusively assigned. We are interested in such distribution $Y$ that the assignment can be done with no residue at all. 
To guarantee an exclusive assignment for a partition $\delta_l$ we introduce $T$ integer variables $\{T_{lb}\}_{b=1}^B$, indicating which of the $\{\beta_b\}_{b=1}^B$  is assigned to it. Therefore, we have
\begin{equation}
\sum_{b=1}^B T_{lb}\beta_b=\delta_l, \quad \sum_{l=1}^L T_{lb}=1, \quad T_{lb} \in \{0,1\}, \quad \text{for all}  \quad  l=1,\dots,L\;\; \text{and}\;\; b=1,\dots,B
\end{equation}
and the optimization objective is simply
\begin{equation}
\sum_{l=1}^L \left(\delta_l-\frac{1}{L}\right)^2 \rightarrow \text{min}.
\end{equation}
Our constraints can easily be added to the mixed integer formulation presented in the previous sections and the new optimization problem is: 
\begin{equation}
\min z^T cc^T z -\frac{2}{L} c^T z
\end{equation}
\begin{equation}
\nonumber
\text{s.t.}\;\; A_{eq} z = b_{eq}
\end{equation}
\begin{equation}
\nonumber
\quad \quad \quad A_{ineq} z \leq b_{ineq}
\end{equation}
\begin{equation}
\nonumber
\quad\;  0 \leq  z \leq 1
\end{equation}
\begin{equation}
\nonumber
\quad\quad\quad\quad\; T_{lb} \in \{0,1\}\;\;  \forall l,b
\end{equation}
where $z$ is a vector of all parameters in our problem  $z=\left[ p_{iab} \,\,  T_{lb}\right]^T$ and $c^T z=\delta$.
\subsection{Minimizing the Number of Shelves }
The problem of minimizing the residue of the assignment, given the number of columns and the number of shelves, may be formulated as a MIQP. In this section we focus on finding the minimal number of shelves $N$ that guarantees zero residue.  Notice that for large enough $N$ the residue goes to zero, as $Y$ tends to take values on a continuous set. We also notice that the residue is a monotonically non-increasing function of $N$, since by allowing a greater number of shelves we can always achieve the same residue by repeating the previous partitioning up to a meaningless split of one of the compartments. These two qualities allow very efficient search methods (gradient, binary etc.) to find the minimal $N$ for which the residue is ``$\epsilon$- close" to zero. 

Here, we suggest the following simple binary search based algorithm for minimizing the number of shelves for a rectangular shaped storing unit:

\begin{enumerate}
\item	Choose a large enough initial value $N$ such that applying it in the MIQP presented above results in zero residue.
\item	Define a step size as $Stp = \lfloor N/2 \rfloor$.
\item	Apply the MIQP with $N'=N-Stp$.
\item	If the residue is zero repeat previous step with $N=N'$ and $Stp=\lfloor Stp/2 \rfloor$. Otherwise repeat the previous step with $N=N'$ and $Stp=-\lfloor Stp/2 \rfloor$. Terminate if $Stp=0$.
\end{enumerate}

\section{Lossless Innovation Representation and its Relation to the Optimal Transportation Problem}

The essence of the lossless innovation representation problem is finding a single marginal distribution  to be matched to multiple ones under varying costs functions. This problem can be viewed as a design generalization of a multi-marginal setup for the well-studied optimal transportation problem \cite{monge1781memoire}.  In other words, we suggest that the optimal transportation problem can be generalized to a design problem in which we are given not a single but multiple source probability measures. Moreover, we are interested not only in finding mappings that minimizes some cost function, but also in finding the single target probability measure that minimizes that cost.

\subsection{The Optimal Transportation Problem}

The optimal transportation problem was presented by \cite{monge1781memoire} and has generated an important branch of mathematics in the last decades. The optimal transportation problem has many applications in multiple fields such as Economics, Physics, Engineering and others. The problem originally studied by Monge was the following: assume we are given a pile of sand (in $\mathbb{R}^3$) and a hole that we have to completely fill up with that sand. Clearly the pile and the hole must have the same volume and different ways of moving the sand will give different costs of the operation. Monge wanted to minimize the cost of this operation. Formally, the optimal transportation problem is defined as follows. Let $X$ and $Y$ be two seperable metric spaces such that any probability measure on $X$ (or $Y$) is a Radon measure. Let $c:X \times Y \rightarrow[0,\infty]$ be a Borel-measurable function. Given probability measure $\mu$ on $X$ and $\nu$ on $Y$, Monge's optimal transportation problem is to find a mapping $T:X \rightarrow Y$ that realizes the infimum 
$$ \inf \left\{ \int_X c(x,T(x))d\mu(x) \bigg| T_{*}(\mu)=\nu  \right\}$$
where $T_{*}(\mu)$ denotes the \textit{push forward} of $\mu$ by $T$. A map $T$ that attains the infimum is called the \textit{optimal transport map}. 

Notice that this formulation of the optimal transportation problem can be ill-posed in some setups where there is no ``one-to-one" transportation scheme. For example, consider the case where the original pile is a Dirac measure but the hole is not shaped in this manner. A major advance in this problem is due to Kantorovich \cite{kantorovich1942translocation} who proposed the notation of a ``weak solution" to the optimal transportation problem; he suggested looking for plans instead of transport maps \cite{kantorovich2006problem}. The main difference between Kantorovich work and Monge formulation is that while the original Monge problem is restricted to transportation of the complete mass at each point on the original pile, the relaxed Kantorovich version allows splitting of masses.   
However, it is clear that no such result can be expected without additional assumptions on the measures and cost. The first existence and uniqueness result is due to Brenier \cite{brenier1987polar}. In his work, Brenier considered the case where both the pile $X$ and the hole $Y$ satisfy $X=Y \in R^n$, and the cost function is  $c(x,y)=|x-y|^2$. Then, he showed that if the probability measure of $X$ is absolutely continuous with respect to the Lebesgue measure there exists a unique optimal transport map. Following \cite{brenier1987polar},  many researchers started working on this problem, showing existence of optimal maps with more general costs. More recently, Pass published a series of papers discussing a multi-marginal generalization of the optimal transportation problem \cite{pass2011uniqueness,pass2012local,pass2013class}. In his work, Pass considered multiple marginal distributions to be matched to a single destination with a given distribution. In his papers, Pass discussed the existence and uniqueness of solutions for both a Monge-like and Kantorovich-like multi-marginal problems, under different measures and cost functions, and the connection between the two formulations.

In our work we generalize the multi-marginal optimal transportation from a design perspective; we look at the multi-marginal optimal transportation problem not only as a minimization problem over a set of mappings but also ask ourselves what is the optimal target measure such that the cost function is minimal. We show that this problem has very broad use in many fields, especially when taking an equivalent form of multiple source measures matched to a single target. More specifically, we focus our interest on a set of mappings that allow unique recovery between the measures. That is, given the source measures and a target measure, one can uniquely recover any realization of the sources from a given realization of the target. This type of mappings hold a special interest in many applications, as it is shown throughout the previous sections.

\section{Discussion}
In this work we introduce a method to represent any stochastic process as an innovation process, under different objectives and constraints. We show that there exists a simple closed-form solution if we allow the outcome process to take values over a continuous set. However, restricting the alphabet size may cause lossy recovery of the original process. Two solutions are presented in the face of two possible objectives in the discrete case. First, assuming the alphabet size is too small to allow lossless recovery, we aim to maximize the mutual information with the original process. Alternatively, we may seek a minimal alphabet size so that a unique recovery is guaranteed, while minimizing the entropy of the resulting process.  In both cases the problem is shown to be hard and several approaches are discussed. In addition, a simple closed-form solution is provided for the binary first order Markov process.  

It is important to mention that our work focuses on sequential analysis of stochastic processes. This means that at every time stamp $k$, we construct a corresponding  innovation process $Y_k$ from $X_k$. This framework may be generalized to a non-sequential setup, where $X^k$ is analyzed as a batch. In fact, this problem is equivalent to the well-studied Independent Component Analysis, both over continuous \cite{hyvarinen1998independent} and discrete \cite{painsky2014generalized,painsky2016Binary} variables. The ICA problem has many applications in learning and inference. Recently, it was further applied to source coding and data compression \cite{painsky2015Universal,painsky2016Simple,painsky2016largetrans,painsky2016compressing,painsky2018lossless}.

The innovation representation problem has many applications, as it simplifies complex systems and allows simpler analytical and computational solutions. In this work we show that the lossless innovation representation may be applied to infer causality, as was previously shown in\cite{ kocaoglu2017entropic} . In addition, we introduce a practical application from the industrial world, denoted as the IKEA problem.

The problem of finding a single marginal distribution function to be fitted to multiple ones under varying costs functions can be viewed as a multi-marginal generalization of the well-studied optimal transportation problem.  In other words, we suggest that the optimal transportation problem can be generalized to a design problem in which we are given not a single but multiple source distribution functions. In this case, we are  interested not only in finding conditional distributions to minimize a cost function, but also in finding the optimal target distribution. We argue that this problem has multiple applications in the fields of Economics, Engineering and others.

\section*{Appendix A}
\label{on_the_uniquness}
For the simplicity of the presentation we reformulate our problem as follows: assume a random variable $Y$ is to be constructed from a random variable $X$ given $X$'s past, denoted as $X_p$.  Therefore, we would like to construct a memoryless random variable $Y=g(X,X_p)$, with a given $F_Y (y)$, such that
\begin{enumerate} [(i)]
\item  $Y$ is statistically independent in $X_p$.
\item  $X$ can be uniquely recovered from $Y$ given $X_p$.
\item $Y \sim F_Y (y)$.
\end{enumerate}
Therefore, our goal is to find such $Y=g(X,X_p)$ and discuss its uniqueness.

Let us first consider a special case where $Y$ is uniformly distributed, $F_Y (y)=y$ for all $y \in [0,1]$. For $Y$ to be statistically independent of $X_p$ it must satisfy
\begin{equation}
F_{Y|X_p} (y | X_p=x_p)=F_Y (y).
\label{bla}
\end{equation}
Expanding the left hand side of (\ref{bla}) we have that for every $x_p$, 
$$F_{Y|X_p} (y | X_p=x_p)=P(Y \leq y | X_p=x_p)=P(g(X,X_p)\leq y | X_p=x_p).$$
The second requirement suggests that $X$ is uniquely recovered from $Y$ and $X_p$, which implies $X=g_{X_p}^{-1} (Y)$. Assume $g(X,X_p)$ is monotonically increasing with respect to $X$. Then, we have that
\begin{align}
\label{5.4}
F_{Y|X_p} (y | X_p=x_p)=&P(g(X,X_p) \leq y| X_p=x_p)=P(X \leq g_{X_p}^{-1} (y) | X_p=x_p)=F_{X|X_p} (g_{X_p}^{-1} (y)|X_p=x_p) 
\end{align}
where the second equality follows from the monotonically increasing behavior of $g(X,X_p)$ with respect to $X$. Therefore, we are looking for a monotonically increasing transformation $x=g_{X_p}^{-1} (y)$ such that
$$ F_{X|X_p} (g_{X_p}^{-1}(y)| X_p=x_p)=F_Y (y)=y.$$
The following lemmas discuss the uniqueness of monotonically increasing mappings when $X$ is a non-atomic (Lemma \ref{uniquness_1}) or atomic (Lemma \ref{uniquenss_2}) measures.

\begin{lemma}
\label{uniquness_1}
Assume $X$ is a non-atomic random variable with a strictly monotonically increasing commutative distribution function  $F_X (x)$ (that is, $X$ takes values on a continuous set). Suppose there exists a transformation on its domain, $x=h(y)$ such that
$$ F_X (x)|_{x=h(y)}=F_Y (y).$$
Then,
\begin{enumerate}[(1)]
\item 	$x=h(y)$ is unique \label{1}
\item 	 $h(y)$ is monotonically non decreasing (increasing, if $F_Y (y)$ is strictly increasing). \label{2}
\end{enumerate}
\end{lemma}
\begin{proof}
Let us begin by proving (\ref{1}). The transformation $x=h(y)$ satisfies
$$ F_X (x) |_{x=h(y)}=F_X (h(y))=P(X\leq h(y))=F_Y (y).$$
Suppose there is another transformation $x=g(y)$  that satisfies the conditions stated above. Then,
$$F_X (x) |_{x=g(y) }=F_X (g(y))=P(X\leq g(y))=F_Y (y).$$
Therefore, 
$$P(X\leq g(y))=P(X\leq h(y)) \quad \forall y.$$
Suppose $h(y)\neq g(y)$. This means that there exists at least a single $y=\tilde{y}$ where $g(\tilde{y})=h(\tilde{y})+\delta$ and $\delta \neq 0$. It follows that
$$ P(X \leq h(\tilde{y})+\delta)=P(X \leq  h(\tilde{y}))$$
or in other words
$$ F_X (h(\tilde{y}))=F_X (h(\tilde{y} )+\delta)$$
which contradicts the monotonically increasing behavior of $F_X (x)$ where the transformation is defined.

As for (\ref{2}), we have that $F_X (h(y))=F_Y (y) $ for all $y$. Therefore, 
$$F_X (h(y+\delta))=F_Y (y+\delta).$$
$F_Y (y)$ is a CDF which means that it satisfies $F_Y (y+\delta) \geq F_Y (y)$. Then, 
$$F_X (h(y+\delta))\geq F_X (h(y))$$
or strictly larger, if $F_Y (y)$ is monotonically increasing.
Since $F_X (x)$ is monotonically increasing we have that $h(y+\delta)\geq h(y)$, or strictly larger, if $F_Y (y)$ is monotonically increasing.  \hfill $\square$
\end{proof}

\begin{lemma}
\label{uniquenss_2}
Assume $X$ is a non-atomic random variable with a commutative distribution function  $F_X (x)$. Suppose that there exists a transformation on its domain, $x=h(y)$ such that
$$ F_X (x)|_{x=h(y)}=F_Y (y).$$
Then,
\begin{enumerate}[(1)]
\item 	$x=h(y)$ is unique up to transformations in zero probability regions in $X$'s domain. \label{3}
\item 	 $h(y)$ is monotonically non-decreasing (increasing, if $F_Y (y)$ is strictly increasing). \label{4}
\end{enumerate}
\end{lemma}

\begin{proof}
(\ref{1}) As in Lemma \ref{uniquness_1}, 	let us assume that there exists another transformation $x= g(y)$  that satisfies the desired conditions. Therefore we have that
$$ P(X\leq g(y))=P(X\leq h(y))\quad \forall y.$$
Assuming $h(y)\neq g(y)$ we conclude that there exists at least a single value $y=\tilde{y}$ such that $g(\tilde{y})=h(\tilde{y})+\delta$ and $\delta \neq 0$. 
If both $h(\tilde{y})$ and $g(\tilde{y} )$ are valid values in $X$'s domain (positive probability) then we have $P(X\leq x_1 )=P(X\leq x_2)$.
This contradicts  $P(X=x_1 )>0$ and $P(X=x_2 )>0$ unless $x_1=x_2$.\\
 Moreover, if $g(\tilde{y} )\in [x_1,x_2 ]$ and $h(\tilde{y}) \notin [x_1,x_2 ]$ then again it contradicts $P(X=x_1 )>0$ and $P(X=x_2 )>0$ unless $x_1=x_2$. The only case in which we are not facing a contradiction is where  $g(\tilde{y}),h(\tilde{y}) \in[x_1,x_2 ]$. In other words, $x= g(y)$  is unique up to transformations in zero probability regions of $X$'s domain (regions which satisfy $P(X=g(\tilde{y} ))=0)$.

\noindent (\ref{4}) The monotonicity proof follows the same derivation as in Lemma \ref{uniquness_1}.    \hfill $\square$

\end{proof}
Therefore, assuming that there exists a transformation $x=g_{X_p}^{-1} (y)$ such that 
$$F_{X|X_p} (g_{X_p}^{-1} (y)|X=x_p)=F_Y (y)=y,$$
then it is unique and monotonically increasing. In this case we have that
\begin{align}
F_Y (y)=&F_{X|X_p} (g_{X_p}^{-1} (y)|X=x_p)=P(X\leq g_{X_p}^{-1} (y)|X=x_p)=P(g(X,X_p)\leq y | X=x_p)  =F_{Y|X_p}(y|X_p=x_p)  
\end{align}
which means $Y$ is statistically independent of $X_p$.
Equivalently, if we find a monotonically increasing transformation $Y=g(X,X_p)$ that satisfies conditions (i), (ii) and (iii) then it is unique.

To this point, we discussed the case in where the functions $g(X,X_p)$ are monotone in $X$. For this set of functions equation (\ref{5.4}) is a sufficient condition for satisfying (i) and (ii). 
However, we may find non monotonically increasing transformations $Y=h(X,X_p)$ which satisfy conditions (i), (ii) and (ii) but do not satisfy (\ref{5.4}). For example: $h(X,X_p)=1-g(X,X_p)$. Notice that these transformations are necessarily measurable, as they map one distribution to another, and reversible with respect to $X$ given $X_p$ (condition ii). In this case, the following properties hold:
\begin{lemma}
\label{uniqueness_3}
Assume h(X,Y) satisfies the three conditions  mentioned above but does not satisfy equation (5.4). Then:
\begin{enumerate}[(1)]
\item $h(X,X_p)$ is not monotonically increasing in $X$ \label{5}
\item $h(X,X_p)$ is necessarily a ``reordering" of $g(X,X_p)$ \label{6}
\end{enumerate}
\end{lemma}

\begin{proof}
(\ref{5}) Assume there exists a transformation $Y=h(X,X_p)$ which satisfy the three conditions (i), (ii) and (iii). Moreover assume $h(X,X_p) \neq g(X,X_p)$.  We know that 
$$F_{Y|X_p} (y|X_p=x_p)=P(h(X,X_p) \leq y | X_p=x_p)=F_Y (y) $$   
but on the other hand, $h(X,X_p) \neq g(X,X_p)$  which implies
$$ F_{X|X_p} (h_{X_p}^{-1} (y)|X_p=x_p) \neq F_Y (y)$$					
since $g(X,X_p)$ is unique. Therefore,
$$ P(h(X,X_p) \leq y | X_p=x_p) \neq P (X \leq h_{X_p}^{-1} (y) | X_p=x_p)$$
which means $h(X,X_p)$ cannot be monotonically increasing.\\

\noindent (\ref{6}) Notice we can always generate a (reversible) transformation of $h(X,X_p)$ that will make it monotonically increasing with respect to $X$, since $X$ is uniquely recoverable from $h(X,X_p)$ and $X_p$. Consider this transformation as $S(h(X,X_p))$. Therefore, we found $Y=S(h(X,X_p))$ such that $y$ is monotonically increasing, independent of $X_p$ and $X$ is uniquely recoverable from $Y$ and $X_p$. This contradicts the uniqueness of $g(X,X_p)$ unless $S(h(X,X_p))=g(X,X_p)$, which means 
$h(X,X_p)=S^{-1} (g(X,X_p))$.  \hfill $\square$
\end{proof}

To conclude, it is enough to find $Y=g(X,X_p)$ which is invertible and monotonically increasing with respect to $X$ given $X_p=x_p$, and satisfies
$$ F_{Y|X_p} (y|X_p=x_p)=F_{Y|X_p} (g_{X_p}^{-1} (y)|X_p=x_p)=F_Y (y)=y.$$ 	
If such $Y=g(X,X_p)$ exists then
\begin{enumerate}
\item  If $F_{X|X_p}(x|X_p=x_p)$ is monotonically increasing, then $Y=g(X,X_p)$ is unique according to Lemma \ref{uniquness_1}
\item  If $X|X_p$ takes on discrete values, then again $Y=g(X,X_p)$ is unique, up to different transformations in zero probability regions of the $X|X_p$
\item  Any other transformations $h(X,X_p)$ that may satisfy conditions (i),(ii) and (iii) is necessarily a function of $g(X,X_p)$ (and not monotonically increasing).
\end{enumerate}

Following lemma \ref{uniquness_1} we define 
$Y=F_{X|X_p}(x|x_p)-\Theta \cdot P_{X|X_p} (x|x_p)$, where $\Theta \sim \text{Unif}[0,1]$ is statistically independent of $X$ and $X_p$. Therefore we have that
\begin{align}
F_{Y|X_p} (y|x_p)=& P(F_{X|X_p} (x|x_p)-\Theta \cdot P_{X|X_p} (x|_p) \leq y | X_p=x_p)=\\\nonumber
&P(F_{X|X_p} (x|x_p)-\Theta \cdot P_{X|X_p} (x|x_p) \leq h^{-1} (y))=y=F_Y(y)
\end{align}
where the first equality follows from the fact that all the terms in $F_{X|X_p} (x|x_p)-\Theta \cdot P_{X|X_p} (x|x_p) \leq h^{-1} (y)$ are already conditioned on $X_p$, or statistically independent of $X_p$, and the second equality follows from $	F_{X|X_p} (x|x_p)-\Theta \cdot P_{X|X_p} (x|x_p) \sim \text{Unif}[0,1]$, according to lemma \ref{uniquness_1}. The third condition is remaining requirement. However, it is easy to see that $Y=F_{X|X_p} (x|x_p)-\Theta \cdot P_{X|X_p} (x|x_p)$ is reversible with respect to $X$ given $X_p=x_p$. Therefore, we found a monotonically increasing transformation $Y=g(X,X_p)$ that satisfies 
$$ F_{X|X_p} (x|x_p)=F_{X|X_p} (g_{X_p}^{-1} (y)|X_p=x_p)=F_Y (y)=y.$$

Going back to our original task, we are interested in finding such $Y=g(X,X_p)$ such that there exists a random variable $Y$ that satisfies conditions (i), (ii) and (iii). Throughout our analysis, we discussed the uniqueness in the case where $Y$ is uniformly distributed. Assume we are now interested in a non-uniformly distributed $Y$. Lemma \ref{uniquness_1} shows that we can always reshape a uniform distribution to any probability measure by applying the inverse of the desired CDF on it. Moreover, if the desired probability measure is non-atomic, this transformation is reversible. Is this mapping unique? This question is already answered by Lemmas \ref{uniquenss_2} and \ref{uniqueness_3}; if we limit ourselves to monotonically increasing transformation, then the solution we found is unique.

However, assume we do not limit ourselves to monotonically increasing transformations and we have a transformation  $V=G(Y)$ that satisfies $V \sim F_V (v)$. Since $Y$  is uniformly distributed we can always shift between local transformations on sets of the same lengths while maintaining the transformation measurable. Then we can always find $S(G(Y))$ which makes it monotonically increasing with respect to $Y$. This contradicts the uniqueness of the monotonically increasing set unless $S(G(Y))$ equals the single unique transformation we found.  

Putting it all together we have a two stage process in which we first generate a uniform transformation and then shape it to a desired distribution $V$ through the inverse of the desired CDF. We show that in both stages, if we limit ourselves to monotonically increasing transformations the solution presented in above is unique. However, if we allow ourselves a broader family of functions we necessarily end up with either the same solution, or a ``reordering" of it which is not monotonically increasing.

\section*{Appendix B}
We would like to show that for $\alpha_1<\alpha_2<\frac{1}{2}$ the following applies:

\begin{equation}
\frac{\alpha_2}{1-\alpha_1-\alpha_2}<\frac{h_b(\alpha_2)-h_b(\alpha_1)+\alpha_2 h_b\left( \frac{\alpha_1}{\alpha_2}\right)}{\alpha_2 h_b\left( \frac{\alpha_1}{\alpha_2}\right)+(1-\alpha_1)h_b\left( \frac{\alpha_2-\alpha_1}{1-\alpha_1}\right)}
\end{equation}
\begin{proof}
Let us first cross multiply both sides of the inequality
\begin{align}
&\alpha^2_2 h_b\left( \frac{\alpha_1}{\alpha_2}\right)+\alpha_2(1-\alpha_1)h_b\left( \frac{\alpha_2-\alpha_1}{1-\alpha_1}\right)<(1-\alpha_1-\alpha_2)(h_b(\alpha_2)-h_b(\alpha_1))+(1-\alpha_1)\alpha_2 h_b\left( \frac{\alpha_1}{\alpha_2}\right)+\alpha^2_2 h_b\left( \frac{\alpha_1}{\alpha_2}\right)
\end{align}
which leads to
\begin{equation}
\nonumber
(1-\alpha_1-\alpha_2)(h_b(\alpha_2)-h_b(\alpha_1))+(1-\alpha_1)\alpha_2 h_b\left( \frac{\alpha_1}{\alpha_2}\right)-\alpha_2(1-\alpha_1)h_b\left( \frac{\alpha_2-\alpha_1}{1-\alpha_1}\right)>0.
\end{equation}
\\
Since $h_b(\alpha_2)-h_b(\alpha_1)>0$ and $1-\alpha_1-\alpha_2>(1-\alpha_1)\alpha_2$ we have that
\begin{align}
\nonumber
&(1-\alpha_1-\alpha_2)(h_b(\alpha_2)-h_b(\alpha_1))+(1-\alpha_1)\alpha_2 h_b\left( \frac{\alpha_1}{\alpha_2}\right)-\alpha_2(1-\alpha_1)h_b\left( \frac{\alpha_2-\alpha_1}{1-\alpha_1}\right)>\\\nonumber
&(1-\alpha_1)\alpha_2\left[h_b(\alpha_2)-h_b(\alpha_1)+h_b\left( \frac{\alpha_1}{\alpha_2}\right)- h_b\left( \frac{\alpha_2-\alpha_1}{1-\alpha_1}\right)  \right].
\end{align}
Therefore, it is enough to show that $h_b(\alpha_2)-h_b(\alpha_1)+h_b\left( \frac{\alpha_1}{\alpha_2}\right)- h_b\left( \frac{\alpha_2-\alpha_1}{1-\alpha_1}\right) >0$. Since $h_b\left( \frac{\alpha_2-\alpha_1}{1-\alpha_1}\right)=h_b\left( \frac{1-\alpha_2}{1-\alpha_1}\right)$ we can rewrite the inequality as
\begin{align}
\nonumber
h_b(\alpha_2)-h_b(\alpha_1)>h_b\left( \frac{1-\alpha_2}{1-\alpha_1}\right)-h_b\left( \frac{\alpha_1}{\alpha_2}\right).
\end{align}
Notice that  $\alpha_1<\alpha_2<\frac{1}{2}$ follows that $\frac{1-\alpha_2}{1-\alpha_1}>\frac{1}{2}$.

 Let us first consider the case where $\frac{\alpha_1}{\alpha_2}\geq\frac{1}{2}$. We have that
\begin{align}
 \frac{1-\alpha_2}{1-\alpha_1}-\frac{\alpha_1}{\alpha_2}=\frac{(\alpha_2-\alpha_1)(1-\alpha_1-\alpha_2)}{(1-\alpha_1)\alpha_2}>0.
\end{align}
Since $\frac{1-\alpha_2}{1-\alpha_1}-\frac{\alpha_1}{\alpha_2}>\frac{1}{2}$ and $h_b(p)$ is monotonically decreasing for $p\geq\frac{1}{2}$, we have that 
\begin{equation}
h_b\left( \frac{1-\alpha_2}{1-\alpha_1}\right)-h_b\left( \frac{\alpha_1}{\alpha_2}\right)<0<h_b(\alpha_2)-h_b(\alpha_1).
\end{equation}
Now consider the case where $\frac{\alpha_1}{\alpha_2}\geq\frac{1}{2}$. We notice that:
\begin{equation}
h_b\left( \frac{1-\alpha_2}{1-\alpha_1}\right)=h_b\left(1- \frac{1-\alpha_2}{1-\alpha_1}\right)=h_b\left( \frac{\alpha_2-\alpha_1}{1-\alpha_1}\right)
\end{equation}
where $\frac{\alpha_2-\alpha_1}{1-\alpha_1}<\frac{1}{2}.$
In addition, 
\begin{equation}
\frac{\alpha_2-\alpha_1}{1-\alpha_1}-\frac{\alpha_1}{\alpha_2}=\frac{(\alpha_2-\alpha_1)^2+\alpha_1(1-\alpha_2)}{(1-\alpha_1)\alpha_2}>0.
\end{equation}
Therefore, we would like to show that 
$$h_b(\alpha_2)-h_b(\alpha_1)>h_b\left( \frac{\alpha_2-\alpha_1}{1-\alpha_1}\right)-h_b\left( \frac{\alpha_1}{\alpha21}\right)$$
where all the binary entropy arguments are smaller than $\frac{1}{2}$ and both sides of the inequality are non-negative. In order to prove this inequality we remember that $h_b(p)$ is monotonically increasing with a decreasing slope, $\frac{\partial}{\partial p}h_b(p)=\log\frac{1-p}{p}$, for $p<\frac{1}{2}$. Then, it is enough to show that $\alpha_1<\frac{\alpha_1}{\alpha_2}$ (immediate result) and 
$$\alpha_2-\alpha_1>\frac{\alpha_2-\alpha_1}{1-\alpha_1}- \frac{\alpha_1}{\alpha_2}.$$
Looking at the difference between the two sides of the inequality we obtain:
\begin{align}
\frac{\alpha_2-\alpha_1}{1-\alpha_1}- \frac{\alpha_1}{\alpha_2}-(\alpha_2-\alpha_1)=&
(\alpha_2-\alpha_1)\frac{\alpha_1}{1-\alpha_1}-\frac{\alpha_1}{\alpha_2}<\frac{1}{2}(1-\alpha_1)\frac{\alpha_1}{1-\alpha_1}-\frac{\alpha_1}{\alpha_2}=\alpha_1\left(\frac{\alpha_2-2}{2\alpha_2} \right)<0
\end{align}
where the inequality follows from $\frac{\alpha_2-\alpha_1}{1-\alpha_1}<\frac{1}{2}$, leading to $\alpha_2-\alpha_1<\frac{1}{2}(1-\alpha_1)$.
\end{proof}

\section*{Appendix C}
We begin with the following proposition:
\begin{proposition}
\label{appendix3_prop1}
Let $X \sim \text{multnom}\left(\gamma_1\dots, \gamma_N\right)$, where $\gamma_i$ are parameters bounded from above such that $\gamma_i \leq b_i$ for all $i=1,\dots,N$ and $\sum b_i \geq 1$ (otherwise a valid solution is infeasible). Further assume $b_i \leq b_{i+1}$ for all $i$. Then,
\begin{enumerate} 
\item $H(X)$ is minimal only if $\gamma_N=b_N$.
\item The minimal entropy is achieved iif there exists $k>0$ such that $\gamma_i=b_i$ for every $i>k$  and $H\left( \frac{\gamma_1}{1-\lambda},\dots,\frac{\gamma_k}{1-\lambda}\right)=0$ where $\lambda=\sum_{i=k+1}^N\gamma_i$.
\end{enumerate}
\end{proposition}
\begin{proof}

First we consider the binary case where $N=2$. Here, the boundary conditions suggest $\gamma_1\leq b_1$ and $\gamma_2\leq b_2$. 
Moreover since $\gamma_2=1-\gamma_1$ we have that $1-b_2 \leq \gamma_1 \leq b_1$.
The minimal entropy is achieved on one of the boundaries,  $\gamma_1=b_1$ or $\gamma_2=b_2$. The resulting entropy is either $H_1 (\gamma_1=b_1 )=H_b (b_1)$  or $ H_2 (\gamma_2=b_2 )=H_b (b_2 )$.
\begin{itemize}
\item For $b_1>\frac{1}{2}$,  we have $H_b (b_2)\leq H_b (b_1)$ since $H_b(X)$ is monotonically decreasing for $X>\frac{1}{2}$.
\item For $b_1<\frac{1}{2}$ and $b_2>\frac{1}{2}$ we have $H_b(b_2)=H_b(1-b_2)\leq H_b(b_1)$ since $1-b_2\leq b_1<\frac{1}{2}$ and $H_b (X)$ is monotonically increasing for $X<\frac{1}{2}$.
\item Assuming $b_1<\frac{1}{2}$ and $b_2<\frac{1}{2}$ we have no feasible solution since $\sum b_i \geq 1$. 
\end{itemize}
Therefore, the minimal entropy is always achieved when $\gamma_2$ is maximal.
Now, consider the case where $N>2$.
Assume $\gamma_N<b_N$ achieves the minimal entropy.
The entropy of $X$ can be written as:
$$H(X)=H_b(\lambda)+\lambda H(\frac{\gamma_i}{\lambda},\frac{\gamma_N}{\lambda})+(1-\lambda)H\left(\frac{\gamma_1}{1-\lambda},\dots, \frac{\gamma_{i-1}}{1-\lambda},\frac{\gamma_{i+1}}{1-\lambda},\dots, \frac{\gamma_{N-1}}{1-\lambda}\right)$$
for any $\gamma_i>0$ and $\lambda=\gamma_i+\gamma_N$.
Looking at $H(\frac{\gamma_i}{\lambda},\frac{\gamma_N}{\lambda})$ we notice we have a constrained two parameters problem (as in the binary case above).
Assuming $b_N \geq \gamma_i+\gamma_N$, the minimum of $H(\frac{\gamma_i}{\lambda},\frac{\gamma_N}{\lambda})$ is  zero since it is a non-constrained minimization problem, and $\tilde{\gamma}_N=\gamma_N+\gamma_i$ achieves lower entropy which contradicts the minimal entropy achieved by $\gamma_N$. See the left chart of Figure \ref{appendix_fig} for example. 

Assuming $b_N<\gamma_i+\gamma_N$,  the minimum of $H(\frac{\gamma_i}{\lambda},\frac{\gamma_N}{\lambda})$ is achieved when $\gamma_N$ is maximal, $\tilde{\gamma}_N=b_N$, as we saw for the $N=2$ case. This again contradicts the minimal entropy achieved by $\gamma_N$. See the right chart of Figure \ref{appendix_fig} for example.

\begin{figure}[h]
\centering
\includegraphics[width = 0.5\textwidth,bb= 120 450 490 640,clip]{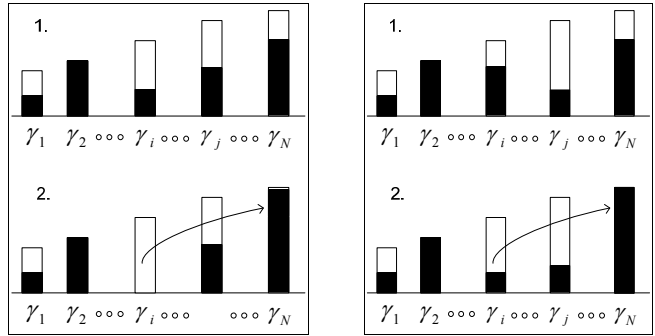}
\caption{Illustrations for entropy maximization in Proposition \ref{appendix3_prop1} }
\label{appendix_fig}
\end{figure}

Therefore, the minimal entropy is achieved only if $\gamma_N=b_N$. Further, setting $\gamma_N$ at its bound, $\gamma_N=b_N$, is a necessary condition for the minimality of $H(X)$. However, now we are facing an $N-1$ parameters entropy minimization problem, for which we also have the same necessary condition of setting $\gamma_{N-1}$ to its upper bound, if possible. Continuing this way, we eventually face a situation in which: 	
\begin{itemize}
\item The upper bound $\gamma_k$ is greater then the rest of the probability not allocated to $\gamma_i>k$. In this case, the minimal entropy of $H\left(\frac{\gamma_1}{1-\lambda},\dots, \frac{\gamma_{i-1}}{1-\lambda},\frac{\gamma_{i+1}}{1-\lambda},\dots, \frac{\gamma_{N-1}}{1-\lambda}\right)$ is just zero, by allocating $1-\lambda$ to a single parameter, $\gamma_k$. 
\item The parameter $\gamma_3$ is set to its upper bound which leaves us with a binary  entropy minimization problem, for which we found the necessary and sufficient condition - we set $\gamma_2$ to its maximal bound.
 \hfill $\square$
\end{itemize}
We now continue to prove Theorem \ref{theorem_lb}.
Following the steps of the previous proof, we start by showing that a necessary condition for minimal entropy is that all lower bounds are satisfied and $\gamma_N$ is at a feasible maxima, $\gamma_N=\min \{ b_N, 1-\sum_{i\neq N}a_i \}$.
We start with the binary case. Here, the boundary conditions suggest $a_1\leq \gamma_1 \leq b_1$ and $a_2 \leq \gamma_2 \leq b_2$. Let us first assume there are no upper bounds (or that the bounds are not forcing any restrictions).
In this case we decide between the two vertices,  $H(\gamma_1=a_1 )$ and $H(\gamma_2=a_2 )$. 
Comparing the two we have:
\begin{itemize}
\item Assuming $a_1>\frac{1}{2}$ we have no feasible solution since $\sum_i a_i>1$.
\item Assuming $a_1<\frac{1}{2}$ and $a_2>\frac{1}{2}$ we have $H_b (a_1 )>H_b (1-a_2)$ iff $a_1>1-a_2$ which is infeasible. Therefore, $H_b (a_1 )\leq H_b (1-a_2 )=H_b (a_2 )$.
\item Assuming $a_1<\frac{1}{2}$ and $a_2<\frac{1}{2}$ we have $H_b (a_1 )\leq H_b (a_2 )$.
\end{itemize}
Therefore, we conclude that  $H(\gamma_1=a_1 ) \leq H(\gamma_2=a_2 )$.
Taking the upper bounds into account we notice that if $b_2>1-a_1$ we may set $\gamma_1=a_1$ and achieve minimal entropy. However, in the case where $b_2<1-a_1$ we cannot set $\gamma_1=a_1$ as it results in a non-valid solution. 
Notice that in this case $\gamma_2=b_2$ is a valid solution since it yields $\gamma_1=1-b_2>a_1$.
Therefore, we compare the other three vertices $H(\gamma_2=a_2 ), H(\gamma_2=b_2 )$ and $H(\gamma_1=b_1 )$.
Comparing  $H(\gamma_2=a_2 )$ with $H(\gamma_2=b_2 )$:

\begin{itemize}
\item Assuming $a_2>\frac{1}{2}$ we have $H(\gamma_2=a_2 )>H(\gamma_2=b_2 )$ since $b_2>a_2 >\frac{1}{2}$.
\item Assuming $a_2<\frac{1}{2}$ we have $H(\gamma_2=a_2 )<H(\gamma_2=b_2 )$ iff  $\frac{1}{2}-a_2>b_2-\frac{1}{2}$ which leads to $b_2<1-a_2$. However, this results in $b_1<1-a_2$ as $b_2<b_1$. Setting $\gamma_2=a_2$ yields $\gamma_1=1-a_2>b_1$ which is not a valid solution.
\end{itemize}
Therefore, $H(\gamma_2=a_2 ) > H(\gamma_2=b_2 )$.

Finally, comparing $H(\gamma_2=b_2)$  with $H(\gamma_1=b_1 )$ results in $H(\gamma_2=b_2 )  < H(\gamma_1=b_1 )$ as we saw in proposition \ref{appendix3_prop1}, assuming $\gamma_2=b_2$ is a valid solution (as shown above).

To conclude, if $b_2>1-a_1$ then the minimal entropy is $H(a_1 )$, achieved at $\gamma_1=a_1$. Otherwise, we set $\gamma_2=b_2$ to achieve a minimal entropy of $H(b_2 )$.

In terms of $\gamma_2$ we have:
\begin{itemize}
\item If $b_2>1-a_1$ we set $\gamma_2=1-a_1$.
\item If $b_2<1-a_1$ we set $\gamma_2=b_2$. 
\end{itemize}
Therefore, we have that $\gamma_2=\min\{b_2,1-a_1\}$.

Let us now consider the case where N>2. Following the steps of Proposition \ref{appendix3_prop1} we notice that we can always decrease the entropy by looking at $\gamma_N$ and choosing $\gamma_i$ such that their binary entropy is not minimal according to the $N=2$ case. Repeating this process results in either $\gamma_N$ cannot be increased anymore (as it achieved its upper bound), or there are no $\gamma_i$ to decrease any further. Therefore, we have $\gamma_N=\min\{b_N,1-\sum_{i \neq N} a_i \}$.

Furthermore, we claim that the minimal entropy is  achieved iif there exists $k>0$ such that 
\begin{itemize}
\item $\gamma_i=b_i$ for all $i>k$.  
\item $\gamma_i=a_i$ for all $i<k$. 
\item $\gamma_k=1-\sum_{i\neq k}\gamma_i$.
\end{itemize}

This is a direct results from the fact that $H(X)$ is minimal only if $\gamma_N=\min\{b_N,1-\sum_{i\neq N} a_i\}$. Therefore, setting $\gamma_N$ at its bound is a necessary condition for the minimality of $H(X)$. 
However, now we are facing an $N-1$ parameters entropy minimization problem, for which we also have the same necessary condition of setting $\gamma_{N-1}$ to its upper bound, if possible. Continuing this way, we eventually face a situation in which:

\begin{itemize}
\item The upper bound of $\gamma_k$ is greater than the rest of the probability not allocated yet.  In this case, the minimal entropy of the $H\left(\frac{\gamma_1}{1-\lambda},\dots, \frac{\gamma_{i-1}}{1-\lambda},\frac{\gamma_{i+1}}{1-\lambda},\dots, \frac{\gamma_{N-1}}{1-\lambda}\right)$ is minimized by setting $\gamma_k$ at its highest, $1-\sum_{i \neq k}\gamma_i$  and the rest of the probabilities at their lowest (since $\gamma_k$ has the highest lower bound of them all).
\item The parameter $\gamma_3$ is set to its upper bound which leaves us with a binary  entropy minimization problem, for which we found the necessary and sufficient condition - $\gamma_2=\min\{b_2,1-a_1\}$. 
\end{itemize}\hfill $\square$
\end{proof}

\section*{Acknowledgment}
We acknowledge Ofer Shayevitz for the discussions that led to the formulation of the innovation representation problem.

\bibliographystyle{IEEEtran}
\bibliography{refs}

\begin{IEEEbiographynophoto}{Amichai Painsky}  (S’12–M’18)
received his B.Sc. in Electrical Engineering from Tel Aviv University (2007), his M.Eng. degree in Electrical Engineering from Princeton University (2009) and his Ph.D. in Statistics from the  School of Mathematical Sciences in Tel Aviv University. He is currently a Post-Doctoral Fellow, co-affiliated with the
Israeli Center of Research Excellence in Algorithms (I-CORE) at
the Hebrew University of Jerusalem, and the Signals, Information and Algorithms
(SIA) Lab at MIT. His research interests include Data Mining, Machine Learning, Statistical Learning and their connection to Information Theory.
\end{IEEEbiographynophoto}

\begin{IEEEbiographynophoto}{Saharon Rosset} is a Professor in the department of Statistics and Operations Research at Tel Aviv University. His research interests are in Computational Biology and Statistical Genetics, Data Mining and Statistical Learning. Prior to his tenure at Tel Aviv, he received his PhD from Stanford University in 2003 and spent four years as a Research Staff Member at IBM Research in New York. He is a five-time winner of major data mining competitions, including KDD Cup (four times) and INFORMS Data Mining Challenge, and two time winner of the best paper award at KDD (ACM SIGKDD International Conference on Knowledge Discovery and Data Mining)
\end{IEEEbiographynophoto}

\begin{IEEEbiographynophoto}{Meir Feder} (S'81-M'87-SM'93-F'99) received the B.Sc and M.Sc degrees
from Tel-Aviv University, Israel and the Sc.D degree from the Massachusetts
Institute of Technology (MIT) Cambridge, and the Woods Hole
Oceanographic Institution, Woods Hole, MA, all in electrical
engineering in 1980, 1984 and 1987, respectively.

After being a research associate and lecturer in MIT he joined
the Department of Electrical Engineering - Systems, School of Electrical Engineering, Tel-Aviv
University, where he is now a Professor and the incumbent of the Information Theory Chair.
He had visiting appointments at the Woods Hole Oceanographic Institution, Scripps
Institute, Bell laboratories and has been a visiting
professor at MIT. He is also extensively involved in the high-tech
industry as an entrepreneur and angel investor.
He co-founded several companies including Peach Networks,
a developer of a server-based interactive TV solution which
was acquired by Microsoft, and Amimon a
provider of ASIC's for wireless high-definition A/V connectivity.
Prof. Feder is a co-recipient of the 1993 IEEE Information Theory
Best Paper Award. He also received the 1978 "creative thinking"
award of the Israeli Defense Forces, the 1994 Tel-Aviv University
prize for Excellent Young Scientists, the 1995 Research Prize of
the Israeli Electronic Industry, and the research prize in applied
electronics of the Ex-Serviceman Association, London, awarded by
Ben-Gurion University.

\end{IEEEbiographynophoto}
\vfill

\end{document}